%% file: paper141.tex
\documentclass[11pt,letterpaper]{article}
\usepackage[margin=1in]{geometry}
\usepackage{mathpazo}
\usepackage{microtype}
\usepackage{graphicx}
\usepackage{verbatim}
\usepackage{xspace}
\usepackage{amsmath}
\usepackage{amssymb}
\usepackage{paralist}
\usepackage{amsthm}
\usepackage[T1]{fontenc}
\usepackage{url}
\usepackage{enumitem}
\usepackage{multirow}
\usepackage{tabularx}
\usepackage[pdftex]{hyperref}
\hypersetup{colorlinks,citecolor=blue}
\usepackage{mathdesign}

\usepackage{fancyhdr}
\fancypagestyle{plain}{%
 \fancyhead{}
 
}
\usepackage[draft,inline,nomargin,index]{fixme}
\fxsetup{mode=multiuser,theme=colorsig,inlineface=\itshape,envface=\itshape}
\FXRegisterAuthor{su}{asu}{Suresh}
\FXRegisterAuthor{sa}{asa}{Samira}
\FXRegisterAuthor{aa}{aaa}{Amir}
\FXRegisterAuthor{cd}{cdr}{Chitradeep}

\usepackage{acronym}

\newcommand{\eps}{\varepsilon}

\newcommand{\reals}{\ensuremath{\mathbb{R}}}
\newcommand{\integers}{\ensuremath{\mathbb{Z}}}
\newcommand{\field}{\mathbb{F}}

\newcommand{\mse}{\textsc{MSE}\xspace}
\newcommand{\sumcheck}{\textsc{SumCheck}\xspace}
\newcommand{\finv}{\textsc{Finv}\xspace}
\newcommand{\subgraph}{\textsc{Subset}\xspace}
\newcommand{\odd}{\ensuremath{\mathbb{O}}}
\newcommand{\disjoint}{\textsc{Disjointness}\xspace}
\newcommand{\spanningtree}{\textsc{r-SpanningTree}\xspace}
\newcommand{\maximality}{\textsc{Maximality}\xspace}

\newtheorem{theorem}{Theorem}[section]
\newtheorem{corollary}{Corollary}[theorem]
\newtheorem{lemma}[theorem]{Lemma}

\title{Streaming Verification of Graph Properties\thanks{This research was supported in part by National Science Foundation under grants IIS-1251049, CNS-1302688}}

\author{Amirali Abdullah\thanks{\noindent Department of Mathematics, University
   of Michigan} \and Samira Daruki\thanks{\noindent School of
   Computing, University of Utah} \and Chitradeep Dutta Roy\thanks{\noindent
   School of Computing, University of Utah} \and Suresh Venkatasubramanian\thanks{\noindent School of Computing, University of Utah}}
   \date{}

\begin{document}
\maketitle
\begin{abstract}
Streaming interactive proofs (SIPs) are a framework for outsourced computation. A computationally limited streaming client (the verifier) hands over a large data set to an untrusted server (the prover) in the cloud and the two parties run a protocol to confirm the correctness of result with high probability.
SIPs are particularly interesting for problems that are hard to solve (or even approximate) well in a streaming setting. The most notable of these problems is finding maximum matchings, which has received intense interest in recent years but has strong lower bounds even for constant factor approximations. 
In this paper, we present efficient streaming interactive proofs that can verify
maximum matchings \emph{exactly}. Our results cover all flavors of matchings
(bipartite/non-bipartite and weighted). In addition, we also present streaming
verifiers for approximate metric TSP. In particular, these are the first
efficient results for weighted matchings and for metric TSP in any streaming verification model.
\end{abstract}

\section{Introduction}
\label{sec:introduction}

The shift from direct computation to outsourcing in the cloud has led to new
ways of thinking about massive scale computation. In the \emph{verification}
setting,  computational effort is split between a computationally weak client
(the \emph{verifier}) who owns the data and wants to solve a desired problem,
and a more powerful server (the \emph{prover}) which performs the computations.
Here the client has only limited (streaming) access to the data, as well as a
bounded ability to talk with the server (measured by the amount of
communication), but wishes to verify the correctness of the prover's answers.
This model can be viewed as a streaming modification of a classic interactive
proof system (a streaming IP, or SIP), and has been the subject of a number of
papers \cite{daruki2015streaming,thaler14semi,cormode2011verifying,chakrabartiverifiable,cormode2013streaming,chakrabarti09annotations,klauck1,klauckprakash} that have established sublinear (verifier) space and communication bounds for classic problems in streaming and data analysis.

In this paper, we present streaming interactive proofs for graph  problems that are traditionally hard for streaming, such as for the maximum matching problem (in bipartite and general graphs, both weighted and unweighted) as well for approximating the traveling salesperson problem. 
In particular, we present protocols that verify a matching \emph{exactly} in a graph using polylogarithmic space and polylogarithmic communication apart from the matching itself. In all our results, we consider the input in the \emph{dynamic streaming model}, where graph edges are presented in arbitrary order in a stream and we allow both deletion and insertion of edges. All our protocols use either $\log n$ rounds of communication or (if the output size is sufficiently large or we are willing to tolerate superlogarithmic communication) constant rounds of communication. 

To prove the above results, we also need SIPs for sub-problems like connectivity, minimum spanning tree and triangle counting. While it is possible to derive similar (and in some cases better) results for these subroutines using known techniques \cite{GKR}, we require explicit protocols that return structures that can be used in the computation pipeline for the TSP. Furthermore, our protocols for these problems are much simpler than what can be obtained by techniques in \cite{GKR}, which require some effort to obtain precise bounds on the size and depth of the circuits corresponding to more complicated parallel algorithms. We summarize our results in Table~\ref{tab:results}. 

\begin{table}[htbp]
\centering
\begin{tabular}{c|c|c|c|c} \hline
 &  \multicolumn{2}{c}{$\log n$ rounds} & \multicolumn{2}{c}{$\gamma = O(1)$ rounds} \\ \cline{2-5}
 Problem   & Verifier Space & Communication & Verifier Space & Communication\\ \hline
Triangle Counting & $\log^2 n$ & $\log^2 n$ & $\log n$ & $n^{1/\gamma}\log n$\\
Matchings (all versions) & $\log^2 n$ & $(\rho + \log n)\log n$ & $\log n$ & $(\rho + n^{1/\gamma'})\log n$ (*)\\
Connectivity & $\log^2 n$ & $n\log n$ & $\log n$ & $n\log n$\\
Minimum Spanning Tree & $\log ^2 n$ & $n\log^2n/\eps$& $\log n $ &$n\log^2 n/\eps$\\
Travelling Salesperson & $\log^2n$ & $n\log^2n/\eps$ & $\log n $ &$n\log^2 n/\eps$
\end{tabular}
\caption{Our Results. All bounds expressed in bits, upto constant factors. For the matching results, $\rho = \min(n, C)$ where $C$ is the cardinality of the optimal matching (weighted or unweighted). Note that for the MST, the verification is for a $(1+\epsilon)$-approximation. For the TSP, the verification is for a $(3/2 + \eps)$-approximation. (*) $\gamma'$ is a linear function of $\gamma$ and is strictly more than $1$ as long as $\gamma$ is a sufficiently large constant.} 
\label{tab:results}
\end{table}

\paragraph*{Significance of our Results.}
\vspace*{-2mm}
While the streaming model of computation has been extremely effective for processing numeric and matrix data, its ability to handle large graphs is limited, even in the so-called \emph{semi-streaming} model where the streaming algorithm is permitted to use space quasilinear in the number of vertices. Recent breakthroughs in graph sketching \cite{mcgregor2014graph} have led to space-efficient approximations for many problems in the semi-streaming model but canonical graph problems like matchings have been shown to be provably hard. 

It is known \cite{kapralov2013better} that no better than a $1-1/e$ approximation to the maximum cardinality matching is possible in the streaming model, even with space $\tilde{O}(n)$. It was also known that even allowing limited communication (effectively a single message from the prover) required a space-communication product of $\Omega(n^2)$ \cite{chakrabarti09annotations, cormode2013streaming}. Our results show that even allowing a few more rounds of communication dramatically improves the space-communication tradeoff for matching, as well as yielding \emph{exact} verification. We note that streaming algorithms for matching vary greatly in performance and complexity depending in whether the graph is weighted or unweighted, bipartite or nonbipartite. In contrast, our results apply to all forms of matching. Interestingly, the special case of \emph{perfect matching}, by virtue of being in RNC \cite{Karp86}, admits an efficient SIP via results by Goldwasser, Kalai and Rothblum \cite{GKR} and Cormode, Thaler and Yi \cite{cormode2011verifying}. Similarly for triangle counting, the best streaming algorithm \cite{ahn2012graph} yields an additive $\eps n^3$ error estimate in polylogarithmic space, and again in the annotation model (effectively a single round of communication) the best result yields a space-communication tradeoff of $n^2\log^2n$, which is almost exponentially worse than the bound we obtain. We note that counting triangles is a classic problem in the sublinear algorithms literature, and identifying optimal space and communication bounds for this problem was posed as an open problem by Graham Cormode in the Bertinoro sublinear algorithms workshop \cite{sublinear}. Our bound for verifying a $3/2+\epsilon$ approximation for the TSP in dynamic graphs is also interesting: a trivial $2$-approximation in the semi-streaming model follows via the MST, but it is open to improve this bound (even on a grid) \cite{tsp-open}. 

In general, our results can be viewed as providing further insight into the tradeoff between space and communication in sublinear algorithms. The annotation model of verification provides $\Omega(n^2)$ lower bounds on the space-communication product for the problems we consider: in that light, the fact that we can obtain polynomially better bounds with only constant number of rounds demonstrates the power of just a few rounds of interaction. We note that as of this paper, virtually all of the canonical hard problems for streaming algorithms (\textsc{Index} \cite{chakrabartiverifiable}, \textsc{Disjointness} \cite{disjoint1,disjoint2}, \textsc{Boolean Hidden Matching} \cite{bhm1,bhm2,bhm4}) admit efficient SIPs. A SIP for \textsc{Index} was presented in  \cite{chakrabartiverifiable}
and we present SIPs for \textsc{Disjointness} and \textsc{Boolean Hidden Matching} here as well.

Our model is also different from a standard multi-pass streaming framework, since communication must remain sublinear in the input and in fact in all our protocols the verifier still reads the input exactly once. 

From a technical perspective, our work continues the \emph{sketching} paradigm for designing efficient graph algorithms. All our results proceed by building linear sketches of the input graph. The key difference is that our sketches are not approximate but algebraic: based on random evaluation of polynomials over finite fields. Our sketches use higher dimensional linearization (``tensorization'') of the input, which might itself be of interest. They also compose: indeed, our solutions are based on building a number of simple primitives that we combine in different ways. Figure \ref{fig:results} illustrates
the interconnections between our tools and results.

\begin{figure}[h!]
\centering
\includegraphics[width=0.5\textwidth]{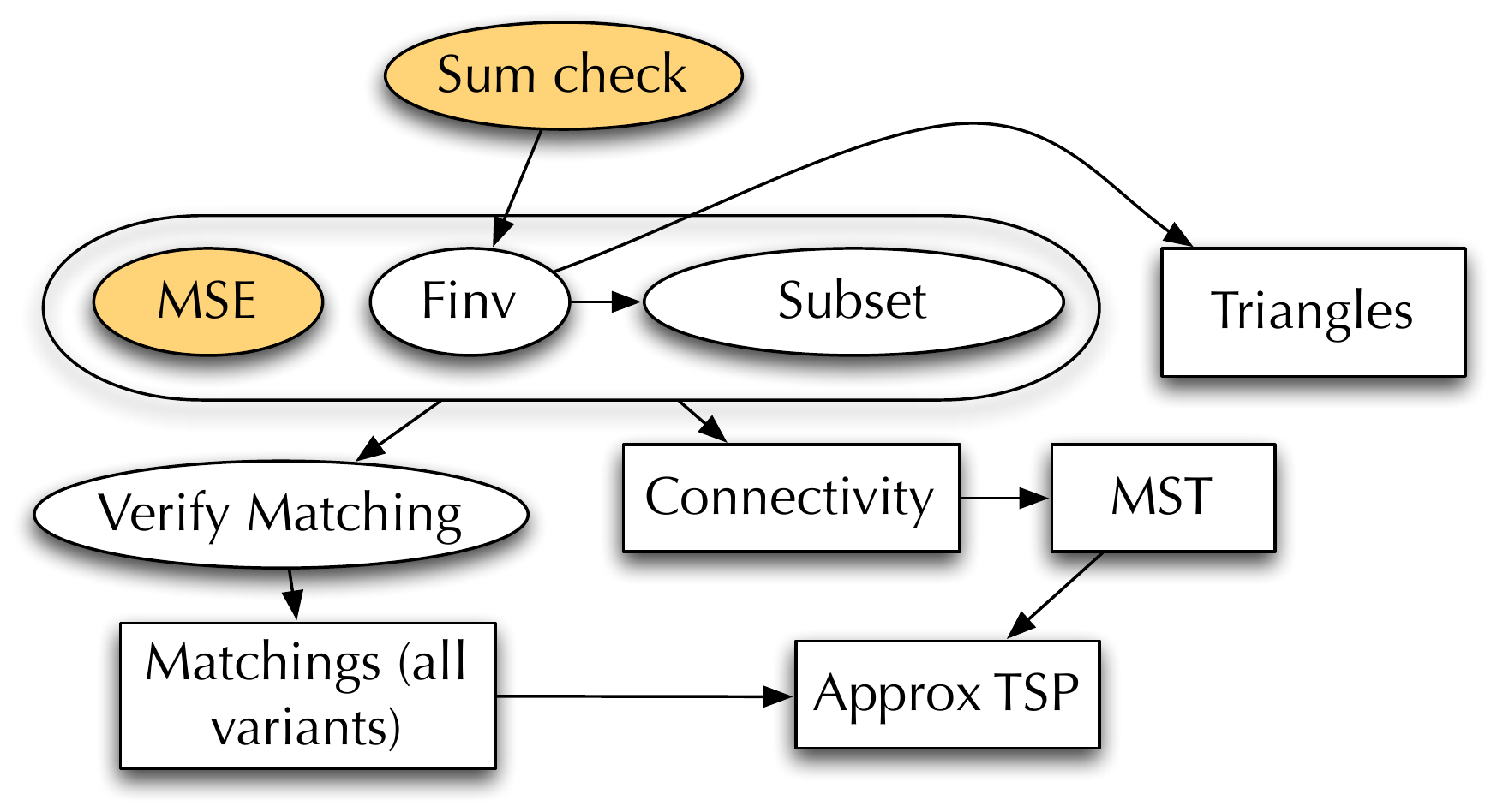}
\caption{Summary of our results. Subroutines are in ovals and problems are in rectangles. Shaded
boxes indicate prior work. An arrow from A to B indicates that B uses A as a subroutine\label{fig:results}}
\end{figure}

\section{Related Work}
\paragraph*{Outsourced computation.}
\vspace*{-2mm}
Work on outsourced computation comes in
three other flavors in addition to SIPs: firstly, there is work on
reducing the verifier and prover complexity without necessarily making the
verifier a sublinear algorithm\cite{GKR,STOC16paper,KRR}, in some cases using
cryptographic assumptions to achieve their bounds. Another approach is the idea of \emph{rational proofs} \cite{Azar-Micali,shikha,guo2016rational,guo2014rational}, in which the verifier uses a
payment function to give the prover incentive to be honest. Moving to sublinear verifiers, there has been research on designing SIPs where the verifier runs in sublinear
\emph{time}\cite{Gur-Rothblum,stoc16-subl-time}.

\paragraph*{Streaming Graph Verification.}
\vspace*{-2mm}
All prior work on streaming graph verification has been in the annotation model,
which in practice resembles a 1-round SIP (a single message from prover to
verifier after the stream has been read). In recent work, Thaler
\cite{thaler14semi} gives protocols for counting triangles, and computing
maximum cardinality matching with both $n \log n$ space and communication
cost. For matching, Chakrabarti \emph{et al.} \cite{chakrabarti09annotations} show that any annotation protocol with space cost $O(n^{1-\delta})$ requires communication cost $\Omega(n^{1+\delta})$ for any $\delta > 0$. They also show that any annotation protocol for graph connectivity with space cost $O(n^{1-\delta})$ requires communication cost $\Omega(n^{1+\delta})$ for any $\delta > 0$.

It is also proved that every protocol for this problem in the annotation model
requires $\Omega(n^2)$ product of space and communication. This is optimal upto
logarithmic factors. Furthermore, they conjecture that achieving smooth
tradeoffs between space and communication cost is impossible, i.e. it is not
known how to reduce the space usage to $o(n \log n)$ without blowing the
communication cost up to $\Omega(n^2)$ or vice versa
\cite{chakrabarti09annotations, thaler14semi}. Note that in all our protocols,
the product of space and communication is $O(n\, \text{poly}\log n)$. 
\paragraph*{Streaming Graph Algorithms.}
\vspace*{-2mm}
In the general dynamic streaming model, $\text{poly}\log 1/\eps$-pass streaming algorithms \cite{ahn2013access, ahn2013linear} give $(1+\eps)$-approximate answers and require $\tilde O(n)$ space In one pass. The best results for matching are \cite{chitnis2015parameterized}  (a parametrized algorithm for computing a maximal matching of size $k$ using $\tilde{O}(nk)$ space) and \cite{assadi2015tight, konrad2015maximum} which give a streaming algorithm for recovering an $n^\epsilon$-approximate maximum matching by maintaining a linear sketch of size $\tilde{O}(n^{2-3\epsilon})$ bits. In the single-pass insert-only streaming model,  Epstein et al. \cite{epstein2011improved} give a constant ($4.91$) factor approximation for weighted graphs using $O(n \log n)$ space. Crouch and Stubbs \cite{crouch2014improved} give a $(4+\epsilon)$-approximation algorithm which is the best known result for weighted matchings in this model. 
Triangle counting in streams has been studied extensively \cite{bar2002reductions, braverman2013hard, buriol2006counting, jowhari2005new, pavan2013counting}. For dynamic graphs, the most space-efficient result is the one by \cite{ahn2012graph}  that provides the aforementioned additive $\eps n^3$ bound in polylogarithmic space. The recent breakthrough in sketch-based graph streaming \cite{ahn2012analyzing} has yielded $\tilde O(n)$ semi-streaming algorithms for computing the connectivity, bipartiteness and minimum spanning trees of dynamic graphs. For
more details, see \cite{mcgregor2014graph}. 

\section{Preliminaries}
\label{sec:background}
We will work in the \emph{streaming interactive proof} (SIP) model first proposed by Cormode et al. \cite{cormode2011verifying}. In this model, there are two players, the prover \textsf{P} and the verifier \textsf{V}. The input consists of a \emph{stream} $\tau$ of items from a universe $\mathcal{U}$. 
Let $f$ be a function mapping $\tau$ to any finite set $\mathcal{S}$. A $k$-message SIP for $f$ works as follows: 
\vspace*{-2mm}
\begin{enumerate}
\item \textsf{V} and $\textsf{P}$ read the input stream and perform some computation on it.
\item \textsf{V} and \textsf{P} then exchange $k$ messages, after which \textsf{V} either outputs a value in $\mathcal{S} \cup \{\perp\}$, where $\perp$ denotes that $\textsc{V}$ is not convinced that the prover followed the prescribed protocol. 
\end{enumerate}
\textsf{V} is randomized. There must exist a prover strategy that causes the verifier to output $f(\tau)$ with probability $1-\eps_c$ for some $\eps_c \leq 1/3$. Similarly, for all prover strategies,  \textsf{V} must output a value in $\{f(\tau), \perp\}$ with probability $1-\eps_s$ for some $\eps_s \leq 1/3$.
The values $\eps_c$ and $\eps_s$ are respectively referred to as the completeness and soundness errors of the protocol. The protocols we design here will have perfect completeness ($\eps_c = 0$) \footnote{The constant $1/3$ appearing in the completeness and soundness requirements is chosen by convention \cite{arorabarak}. The constant $1/3$ can be replaced with any other constant in $(0, 1)$ without affecting the theory in any way.} We note that the annotated stream model of Chakrabarti et al. \cite{chakrabarti09annotations} essentially corresponds to one-message SIPs.\footnote{Technically, the annotated data streaming model allows the annotation to be interleaved with the stream updates, while the SIP model does not allow the prover and verifier to communicate until after the stream has passed. However, almost all known annotated data streaming protocols do not utilize the ability to interleave the annotation with the stream, and hence are actually 1-message SIPs, but without any interaction from the verifier to prover side.}

\paragraph*{Input Model.}
We will assume the input is presented as stream updates to a vector. In general, each element of this stream is a tuple $(i, \delta)$, where each $i$ lies in a universe $\mathcal{U}$ of size $u$, and $\delta \in \{+1, -1\}$. The data stream implicitly defines a frequency vector $\mathbf{a}=(a_1, \dots, a_u)$, where $a_i$ is the sum of all $\delta$ values associated with $i$ in the stream. The stream update $(i, \delta)$ is thus the implicit update $\mathbf{a}[i] \leftarrow \mathbf{a}[i] + \delta$. In this paper, the stream consists of edges drawn from $\mathcal{U} = [n] \times [n]$ along with weight information as needed. As is standard, we assume that edge weights are drawn from $[n^c]$ for some constant $c$. We allow edges to be inserted and deleted  but the final edge multiplicity is $0$ or $1$, and also mandate that the length of the stream is polynomial in $n$. Finally, for weighted graphs, we further constrain that the edge weight updates be atomic, i.e., that an edge along with its full weight be inserted or deleted at each step.

There are three parameters that control the complexity of our protocols: the vector length $u$, the length of stream $s$ and the maximum size of a coordinate $M = max_i \mathbf{a}_i$. In the protocols discussed in this paper $M$ will always be upper bounded by some polynomial in $u$, i.e. $\log M = O(\log u)$. All algorithms we present use linear sketches, and so the stream length $s$ only affects verifier running time. In Lemma \ref{verifier-time} we discuss how to reduce even this dependence, so that verifier update time becomes polylogarithmic on each step.

\paragraph*{Costs.}
\label{sec:costs}
A SIP has two costs: the verifier space, and the total communication, expressed as the number of bits exchanged between \textsf{V} and \textsf{P}. We will use the notation $(A, B)$ to denote a SIP with verifier space $O(A)$ and total communication  $O(B)$. We will also consider the number of rounds of communication between \textsf{V} and \textsf{P}. The basic versions of our protocols will require $\log n$ rounds, and we later show how to improve this to a constant number of rounds while maintaining the same space and similar communication cost otherwise.

\section{Overview of our Techniques}
\label{sec:overview}

For all the problems that we discuss the input is a data stream of edges of a graph where for an
edge $e$ an element in the stream is of the form $(i, j, \Delta)$. Now all our
protocols proceed as follows. We define a domain $\mathcal{U}$ of size $u$ and a
frequency vector $\mathbf{a} \in \integers^u$ whose entries are indexed by
elements of $\mathcal{U}$. A particular protocol might define a number of such
vectors, each over a different domain. Each stream element will trigger a set of
indices from $\mathcal{U}$ at which to update $\mathbf{a}$.
For example in case of matching, we derive this constraint universe from the LP
certificate, whereas for counting triangles our universe is derived from all
$O(n^3)$ possible three-tuples of the vertices.

The key idea in all our protocols is that since we cannot maintain $\mathbf{a}$ explicitly due to limited space, we instead maintain a linear \emph{sketch} of $\mathbf{a}$ that varies depending on the problem being solved. This sketch is computed as follows. We will design a polynomial that acts as a \emph{low-degree} extension of $f$ over an extension field $\field$ and can be written as $p(x_1, \ldots, x_d) = \sum_{u \in \mathcal{U}} a[u] g_u(x_1, x_2, \ldots, x_d)$. The crucial property of this polynomial is that it is \emph{linear} in the entries of $\mathbf{a}$. This means that polynomial evaluation at any fixed point $\mathbf{r} = (r_1, r_2, \ldots, r_d)$ is easy in a stream: when we see an update $a[u] \leftarrow a[u] + \Delta$, we merely need to add the expression $\Delta g_u(\mathbf{r})$ to a running tally. Our sketch will always be a polynomial evaluation at a \emph{random} point $\mathbf{r}$. Once the stream has passed, \textsf{V} and the prover \textsf{P} will engage in a conversation that might involve further sketches as well as further updates to the current sketch. In our descriptions, we will use the imprecise but convenient shorthand ``increment $\mathbf{a}[u]$'' to mean ``update a linear sketch of some low-degree extension of a function of $\mathbf{a}$''. It should be clear in each context what the specific function is. 

As mentioned earlier,  a single stream update of the form $(i, j, \Delta)$ might trigger updates in many entries of $\mathbf{a}$, each of which will be indexed by a multidimensional vector. We will use the wild-card symbol '$*$' to indicate that all values of that coordinate in the index should be considered. For example, suppose $\mathcal{U} \subseteq [n] \times [n] \times [n]$. The instruction ``update $\mathbf{a}[(i, *, j)]$'' should be read as ``update all entries $\mathbf{a}[t]$ where $t \in \{ (i, s, j) \mid s \in [n], (i,s,j) \in \mathcal{U}\}$''. We show later how to do these updates implicitly, so that verifier time remains suitably bounded.

\section{Some Useful Protocols}
\label{sec:protocols}
We will make use of two basic tools in our algorithms: Reed-Solomon fingerprints for testing vector equality, and the streaming \sumcheck protocol of Cormode et al. \cite{cormode2011verifying}. We summarize the main properties of these protocols here: for more details, the reader is referred to the original papers. 

\paragraph*{Multi-Set Equality (\mse).}
We are given streaming updates to the entries of two vectors $\mathbf{a}, \mathbf{a}' \in \integers^u$ and wish to check $\mathbf{a} = \mathbf{a}'$. 
Reed-Solomon fingerprinting is a standard technique to solve \mse  using only logarithmic space. 

\begin{theorem}[\mse, \cite{cormode2013streaming}]
\label{thm:multi-set-equality}
Suppose we are given stream updates to two vectors $\mathbf{a}, \mathbf{a}' \in \integers^u$ guaranteed to satisfy  $|\mathbf{a}_i|, |\mathbf{a}'_i| \leq M$ at the end of the data stream. Let $t = \max(M, u)$.
There is a streaming algorithm using $O(\log t)$ space, satisfying the following properties: (i) If $\mathbf{a} = \mathbf{a'}$, then the streaming algorithm outputs 1 with probability 1. (ii) If $\mathbf{a} \neq \mathbf{a'}$, then the streaming algorithm outputs 0 with probability at least $1-1/t^2$. 
\end{theorem}

\paragraph*{The \sumcheck Protocol.}
We are given streaming updates to a vector $\mathbf{a} \in \integers^u$ and a univariate polynomial  $h \colon \integers \to \integers$. The \textsc{Sum Check} problem (\sumcheck) is to verify a claim that $\sum_i h(\mathbf{a}_i) = K$. 

\begin{lemma}[\sumcheck, \cite{cormode2011verifying}]
\label{lem::sum-check-protocol}
There is a SIP to verify that $\sum_{i \in [u]} h(\mathbf{a}_i) = K$ for some claimed $K$.  The total number of rounds is $O(\log u)$ and the cost of the protocol is $\left( \log(u) \log|\field|  , \text{deg}(h) \log(u)   \log|\field|  \right)$. 
\end{lemma}

Here are the two other protocols that act as building blocks for our graph verification protocols.

\paragraph*{Inverse Protocol (\finv).}
 Let $\mathbf{a} \in \integers^u$ be a (frequency) vector. The \emph{inverse frequency} function $F^{-1}_k$ for a fixed $k$ is the number of elements of $\mathbf{a}$ that have frequency $k$: $F^{-1}_k(\mathbf{a}) = |\{ i \mid \mathbf{a}_i = k\}|$. Let $h_k(i) = 1$ for $i=k$ and $0$ otherwise. We can then define $F^{-1}_k(\mathbf{a}) = \sum_i h_k(\mathbf{a}_i)$. Note that the domain of $h_k$ is $[M]$ where $M = \max_i \mathbf{a}_i$. We will refer to the problem of verifying a claimed value of $F^{-1}_k$ as \finv.
By using Lemma~\ref{lem::sum-check-protocol}, there is a simple SIP for \finv. We restate the related results here \cite{cormode2011verifying}. 

\begin{lemma}[\finv, \cite{cormode2011verifying}]
\label{lemma:finv}
Given stream updates to a vector $\mathbf{a} \in \integers^u$ such that $\max_i \mathbf{a}_i = M$ and a fixed integer $k$ there is a SIP to verify the claim $F^{-1}_k(\mathbf{a}) = K$ with cost $(\log^2 u, M\log^2 u)$ in $\log u$ rounds. 
\end{lemma}

\textbf{Remark 1} Note that the same result holds if instead of verifying an inverse query for a single frequency $k$, we wish to verify it for a set of frequencies. Let $S \subset [M]$ and let $F^{-1}_S = |\{ i | \mathbf{a}_i \in S\}|$. Then using the same idea as above, there is a SIP for verifying a claimed value of $F^{-1}_S$ with costs given by Lemma~\ref{lemma:finv}. 

\textbf{Remark 2} Note that in the protocols presented in this paper later, the input to the \finv is not the graph edges itself, but instead the \finv is applied to the derived stream updates triggered by each input stream elements. As stated before, a single stream update of the form $(i, j, \Delta)$ might trigger updates in many entries of vector $\mathbf{a}$, which is defined based on the problem.

\paragraph*{\subgraph Protocol.}
We now present a new protocol for a variant of the vector equality test described in
Theorem~\ref{thm:multi-set-equality}. While this problem has been studied in the
annotation model, it requires space-communication product of $\Omega(u^2)$ communication in that setting. 

\begin{lemma}[\subgraph]
\label{lem:subgraph}
Let $E \subset [u]$ be a set of elements, and let $S \subset [u]$ be another set owned by \textsf{P}. There is a SIP to verify a claim that $S \subset E$ with cost $(\log^2 u, (|S| + \log u) \log u)$ in $\log u$
rounds.
\end{lemma}
\begin{proof}
Consider a vector $\mathbf{\bar{a}}$ with length $u$, in which the verifier does the following updates: for each element in set $E$, increment the corresponding value in vector $\mathbf{\bar{a}}$ by $+1$ and for each element in set $S$, decrements the corresponding value in vector $\mathbf{\bar{a}}$ by $-1$. Let the vector $\mathbf{a} \in \{0,1\}^u$ be the characteristic vector of $E$, and let $\mathbf{a'}$ be the characteristic vector of $S$. Thus, $\mathbf{\bar{a}} = \mathbf{a} - \mathbf{a'}$. By applying $F^{-1}_{-1}$ protocol on $\mathbf{\bar{a}}$, verifier can determine if $S \subset E$ or not. Note that in vector $\mathbf{\bar{a}}$, $M = 1$. Then the protocol cost follows by Lemma \ref{lemma:finv}.
\end{proof}

\input{triangles-appendix}
\input{verifier-time}
\section{SIP for MAX-MATCHING in Bipartite Graphs}
\label{sec:mcmb}
We now  present a SIP for maximum cardinality matching in bipartite graphs. The prover \textsf{P} needs to generate two certificates: an actual matching, and a proof that this is optimal. By K\"{o}nig's theorem \cite{konig}, a bipartite graph has a maximum matching of size  $k$ if and only if it has a minimum vertex cover of size $k$. Therefore, \textsf{P}'s proof consists of two parts:
\begin{inparaenum}[a)]
\item Send the claimed optimal matching $M\subset E$ of size $k$
\item Send a vertex cover $S\subset V$of size $k$. 
\end{inparaenum}
\textsf{V} has three tasks:
\begin{inparaenum}[i)]
\item Verify that $M$ is a matching and that $M \subset E$.
\item Verify that $S$ covers all edges in $E$. 
\item Verify that $|M| = |S|$.
\end{inparaenum} We describe protocols for first two tasks and the third task is trivially solvable by counting the length of the streams and can be done in $\log n$ space. \textsf{V} will run the three protocols in parallel.

\paragraph*{Verifying a Matching.}
\label{sec:verifying-matching}
Verifying that $M \subset E$ can be done by running the \subgraph protocol from Lemma~\ref{lem:subgraph} on $E$ and the claimed matching $M$. A set of edges $M$ is a matching if each vertex has degree at most $1$ on the subgraph defined by $M$. Interpreted another way, let $\tau_M$ be the stream of endpoints of edges in $M$. Then each item in $\tau_M$ must have frequency $1$. This motivates the following protocol, based on Theorem \ref{thm:multi-set-equality}. \textsf{V} treats $\tau_M$ as a sequence of updates to a frequency vector $\mathbf{a} \in \integers^{|V|}$ counting the number of occurrences of each vertex. \textsf{V} then asks \textsf{P} to send a stream of all the vertices incident on edges of $M$ as updates to a different frequency vector $\mathbf{a}'$. \textsf{V} then runs the \mse protocol to verify that these are the same. 

\paragraph*{Verifying that $S$ is a Vertex Cover.}
The difficulty with verifying a vertex cover is that \textsf{V} no longer has streaming access to $E$. However, we can once again reformulate the verification in terms of frequency vectors. $S$ is a vertex cover if and only if each edge of $E$ is incident to some vertex in $S$. Let $\mathbf{a}, \mathbf{a}' \in \integers^{\binom{n}{2}}$ be  vectors  indexed by $\mathcal{U} = \{(i,j), i, j \in V, i  < j\}$. On receiving the input stream edge $e = (i,j, \Delta), i < j$, \textsf{V} increments $\mathbf{a}[(i,j)]$ by $\Delta$.

For each vertex $i \in S$ that \textsf{P} sends, we increment all entries $\mathbf{a}'[(i, *)]$ and $\mathbf{a}'[(*, i)]$. Now it is easy to see that $S$ is a vertex cover if and only there are no entries in $\mathbf{a - a'}$ with value $1$ (because these entries correspond to edges that have \emph{not} been covered by a vertex in $S$). This yields the following verification protocol. 
\begin{enumerate}
\item \textsf{V} processes the input edge stream for the $F^{-1}_{1}$ protocol, maintaining updates to a vector $\mathbf{a}$. 
\item \textsf{P} sends over a claimed vertex cover $S$ of size $c^*$ one vertex at a time.  For each vertex $i \in S$, \textsf{V} \emph{decrements} all entries $\mathbf{a}[(i, *)]$ and $\mathbf{a}[(*, i)]$.
\item \textsf{V} runs \finv to verify that $F^{-1}_1(\mathbf{a}) = 0$. 
\end{enumerate}

The bounds for this protocol follow from Lemmas \ref{lemma:finv}, \ref{lem:subgraph} and Theorem \ref{thm:multi-set-equality}: 
\begin{theorem}
Given an input bipartite graph with $n$ vertices, there exists a streaming interactive protocol  for verifying the maximum-matching with $\log n$ rounds of communication, and cost $(\log^2 n,   (c^* + \log n) \log n)$, where $c^*$ is the size of the optimal matching.
\end{theorem}

\input{mwmb}

\section{SIP for Maximum-Weight-Matching in General Graphs}
\label{sec:mwmg}
We now turn to the most general setting: of maximum weight matching in general graphs.
This of course subsumes the easier case of maximum cardinality matching in general graphs, and while there is a slightly simpler protocol for that problem based on the Tutte-Berge characterization of maximum cardinality matchings \cite{tutte,berge}, we will not discuss it here.

We will use the odd-set based LP-duality characterization of maximum weight matchings due to Cunningham and Marsh. 
Let $\mathbb{O}(V)$ denote the set of all odd-cardinality subsets of $V$
Let $y_i \in [n^c]$ define non-negative integral weight on vertex $v_i$,   $z_U \in [n^c]$ define a non-negative integral weight on an \emph{odd-cardinality} subset $U \in \mathbb{O}(V)$,  $w_{ij} \in [n^c]$ define the weight of an edge $e = (i,j)$ and $c^* \in [n^{c+1}]$ be the weight of a maximum weight matching on $G$. 
We define $y$ and $z$ to be \emph{dual feasible} if $y_i+ y_j + \sum_{\substack{ U \in \odd(V)  \\ i,j \in  U }   } z_U \geq w_{i,j}, \forall i,j $

A collection of sets is said to be \emph{laminar}, if any two sets in the collection are either disjoint or nested (one is contained in the other). Note that such a family must have size linear in the size of the ground set. Standard LP-duality and the Cunningham-Marsh theorem state that:

\begin{theorem}[\cite{cm}]
\label{thm:cm}
For every integral set of edge weights $W$, and choices of dual feasible integral vectors $y$ and $z$,
$c^* \leq \sum_{v\in V} y_v +\sum_{U \in \odd(V)} z_U  \left \lfloor \frac{1}{2} |U| \right \rfloor.$
Furthermore, there exist vectors $y$ and $z$ that are dual feasible such that $\{ U : z_U > 0 \}$ is laminar and for which  the above upper bound achieves equality.
\end{theorem}

We design a protocol that will verify that each dual edge constraint is satisfied by the dual variables. The laminar family $\{ U : z_U > 0 \}$ can be viewed as a collection of nested subsets (each of which we call a \emph{claw}) that are disjoint from each other. Within each claw, a set $U$ can be described by giving each vertex $v$ in order of increasing \emph{level} $\ell(v)$: the number of sets $v$ is contained in (see Figure \ref{fig:laminar}).
\begin{figure}[h!]
\centering
\includegraphics[width=0.5\textwidth]{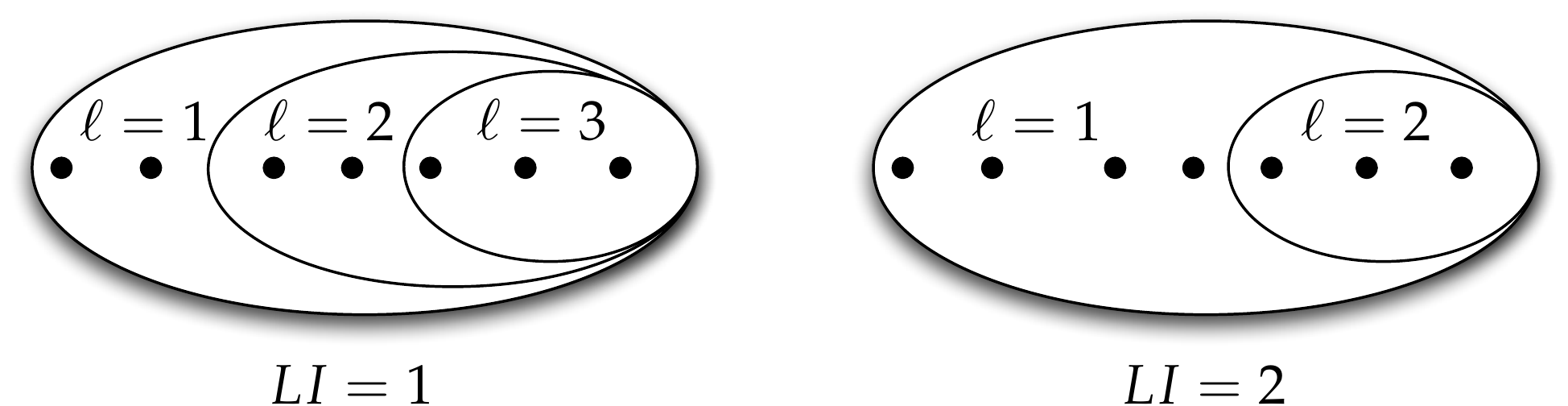}
\caption{A Laminar family\label{fig:laminar}}
\end{figure}
The prover will describe a set $U$ and its associated $z_U$ by the tuple $(LI, \ell, r_U, \partial U)$, where $1 \leq LI \leq n$ is the index of the claw $U$ is contained in, $\ell = \ell(U)$, $r_U = \sum_{U' \supseteq U'} z_{U'}$ and $\partial U = U \setminus \cup_{U'' \subset U} U''$. For an edge $e = (i,j)$ let $r_e = \sum_{i,j \in U, U \in \odd(V)} z_U$ represent the weight assigned to an edge by weight vector $z$ on the laminar family. Any edge whose endpoints lie in different claws will have $r_e = 0$. For a vertex $v$, let $r_v  = \min_{v \in U} r_U$. For an edge $e = (v,w)$ whose endpoints lie in the same claw, it is easy to see that $r_e = \min(r_v, r_w)$, or equivalently that $r_e = r_{\arg \min (\ell(v), \ell(w))}$. For such an edge, let  $\ell_{e, \downarrow} = \min (\ell(u), \ell(v))$ and $\ell_{e, \uparrow} = \max (\ell(u), \ell(v))$. We will use $LI(e) \in [n]$ to denote the index of the claw that the endpoints of $e$ belong to. 
\paragraph*{The Protocol.}
\textsf{V} prepares to make updates to a vector $\mathbf{a}$ with entries indexed by $\mathcal{U} = \mathcal{U}_1 \cup \mathcal{U}_2$.  $\mathcal{U}_1$ consists of all tuples of the form $ \{ (i,j, w, y, y', LI, \ell, \ell', r)\}$ and $\mathcal{U}_2$ consists of all tuples of the form $ \{ (i,j,w, y, y',0,0,0,0) \}$
 where $i < j, i, j, LI, \ell, \ell' \in [n]$, $y, y',r, w \in [n^c]$ and tuples in $\mathcal{U}_1$ must satisfy  1) $w \leq y + y' + r$ and 2) it is \emph{not} simultaneously true that $y + y' \geq w$ and $r > 0$. Note that  $\mathbf{a} \in \integers^u$ where $u = O(n^{4c + 5})$ and all weights are bounded by $n^c$. 
\begin{enumerate}
\item \textsf{V} prepares to process the stream for an $F^{-1}_5$ query. When $\textsf{V}$ sees an edge update of form $(e, w_e, \Delta)$, it updates all entries  $\mathbf{a}[(e, w_e, *, *, *, *, *, *)]$. 
\item \textsf{P} sends a list of vertices $(i, y_i)$ in order of increasing $i$. For each $(i, y_i)$, \textsf{V} increments by $1$ the count of all entries $\mathbf{a}[(i, *, *, y_i, *, *, * , * , *]$ and  $\mathbf{a}[(*,i,*, *, y_i, *,*, *,*)]$ with indices drawn from $\mathcal{U}_1$. Note that \textsf{P} only sends vertices with nonzero weight, but since they are sent in increasing order, \textsf{V} can infer the missing entries and issue updates to $\mathbf{a}$ as above. \textsf{V} also maintains the sum of all $y_i$. 
\item \textsf{P} sends the description of the laminar family in the form of tuples $(LI, \ell, r_U, \partial U)$, sorted in lexicographic order by $LI$ and then by $\ell$. \textsf{V} performs the following operations.
  \begin{enumerate}
  \item \textsf{V} increments all entries of the form $(i, *,*, y_i, *, 0, 0 , 0 , 0)$ or $(*,i,*, *, y_i, 0,0, 0,0)$ by $2$ to account for edges which are satisfied by only vector $y$.
  \item \textsf{V} maintains the sum $\Sigma_R$ of all $r_U$ seen thus far. If the tuple is deepest level for a given claw (easily verified by retaining a one-tuple lookahead) then \textsf{V} adds $r_U$ to another running sum $\Sigma_{\max}$. 
  \item \textsf{V} verifies that the entries appear in sorted order and that $r_U$ is monotone increasing. 
  \item \textsf{V} updates the fingerprint structure from Theorem \ref{thm:multi-set-equality} with each vertex in $\partial U$. 
  \item  For each $v \in \partial U$, \textsf{V} increments (subject to our two constraints on the universe) all entries of $\mathbf{a}$ indexed by tuples of the form $(e, w_e, *, *, LI, *, \ell, *)$ and all entries indexed by tuples of the form $(e, w_e, *, *, LI, \ell, *, r_U)$, where $e$ is any edge containing $v$ as an endpoint. 
\item \textsf{V} ensures all sets presented are odd by verifying that for each $LI$, all $|\partial U|$ except the last one are even. 
  \end{enumerate}
\item \textsf{P} sends \textsf{V} all vertices participating in the laminar family in ascending order of vertex label. \textsf{V} verifies that the fingerprint constructed from this stream matches the fingerprint constructed earlier, and hence that all the claws are disjoint. 
\item \textsf{V} runs a verification protocol for $F^{-1}_5(\mathbf{a})$ and accepts if $F^{-1}_5(\mathbf{a}) = m$, returning $\Sigma_r$ and $\Sigma_{\max}$.
\end{enumerate}

Define $c^s$ as the certificate size, which is upper bounded by the matching cardinality. Then:
\begin{theorem}
  Given dynamic updates to a weighted graph on $n$ vertices with all weights bounded polynomially in $n$, there is a SIP with cost $(\log^2 n, (c^s + \log n) \log n)$, where $c^s$ is the cardinality of maximum matching, that runs in $\log n$ rounds and verifies the size of a maximum weight matching.
\end{theorem}
\begin{proof}
In parallel, \textsf{V} and \textsf{P} run protocols to verify a claimed matching as well as its optimality. The correctness and resource bounds for verifying the matching follow from Section~\ref{sec:mcmb}. We now turn to verifying the optimality of this matching. 
  The verifier must establish the following facts:
  \begin{inparaenum}[(i)]
  \item \textsf{P} provides a valid laminar family of odd sets.
  \item The lower and upper bounds are equal. 
  \item All dual constraints are satisfied.
  \end{inparaenum}

Since the verifier fingerprints the vertices in each claw and then asks \textsf{P} to replay all vertices that participate in the laminar structure, it can verify that no vertex is repeated and therefore that the family is indeed laminar. Each $\partial U$ in a claw can be written as the difference of two odd sets, except the deepest one (for which $\partial U = U$). Therefore, the cardinality of each $\partial U$ must be even, except for the deepest one. \textsf{V} verifies this claim, establishing that the laminar family comprises of odd sets. 

Consider the term $\sum_U z_U\lfloor |U|/2 \rfloor$ in the dual cost. Since each $U$ is odd, this can be rewritten as $(1/2)( \sum_u z_u |U| - \sum_U z_U)$. Consider the odd sets $U_0 \supset U_1 \supset \ldots \supset U_l$ in a single claw. We have $r_{U_j} = \sum_{i \le j} z_{U_i}$, and therefore $\sum_j r_{U_j} = \sum_j \sum_{i\le j} z_{U_i}$. Reordering, this is equal to $\sum_{i \le j} \sum_j z_{U_i} = \sum_i z_{U_i} |U_i|$. Also, $r_{U_l} = \sum_i z_{U_i}$. Summing over all claws, $\Sigma_r = \sum_U z_U |U|$ and $\Sigma_{\max} = \sum_U z_U$. Therefore, $\sum_i y_i + \Sigma_r - \Sigma_{\max}$ equals the cost of the dual solution provided by \textsf{P}. 

Finally we turn to validating the dual constraints. Consider an edge $e = (i,j)$ whose dual constraints are satisfied: i.e. \textsf{P} provides $y_i, y_j$ and $z_U$ such that $y_i + y_j + r_{ij} \ge w_e$. Firstly, consider the case when $r_{ij} > 0$. In this case, the edge belongs to some claw $LI$. Let its lower and upper endpoints vertex levels be $s, t$, corresponding to odd sets $U_S, U_t$. Consider now the entry of $\mathbf{a}$ indexed by $(e, y_i, y_j, LI, s, t, r_{ij})$.  This entry is updated when $e$ is initially encountered and ends up with a net count of $1$ at the end of input processing. It is incremented twice when \textsf{P} sends the $(i, y_i)$ and $(j, y_j)$. When \textsf{P} sends $U_s$ this entry is incremented because $r_{ij} = r_{U_s} = \min (r_{U_S}, r_{U_t})$ and when \textsf{P} sends $U_t$ this entry is incremented because $U_t$ has level $t$, returning a final count of $5$. If $r_{ij} = 0$ (for example when the edge crosses a claw), then the entry indexed by $(e, w_e, y_i, y_j, 0,0,0,0)$ is incremented when $e$ is read. It is not updated when \textsf{P} sends $(i, y_i)$ or $(j, y_j)$. When \textsf{P} sends the laminar family, \textsf{V} increments this entry by $2$ twice (one for each of $i$ and $j$) because we know that $y_i + y_j \ge w_e$.  In this case, the entry indexed by $(i,j,w_e,y_i,y_j,0,0,0,0)$ will be exactly $5$. Thus, for each satisfied edge there is exactly one entry of $\mathbf{a}$ that has a count of $5$.

Conversely, suppose $e$ is not satisfied by the dual constraints, for which a necessary condition is that $y_i + y_j < w_e$. Firstly, note that any entry indexed by $(i,j,w_e, *, *, 0,0,0,0)$ will receive only two increments: one from reading the edge, and another from one of $y_i$ and $y_j$, but not both. Secondly, consider any entry with an index of the form $(i,j,w_e, *, *, LI, *, *, *)$ for $LI > 0$. Each such entry gets a single increment from reading $e$ and two increments when \textsf{P} sends $(i, y_i)$ and $(j, y_j)$. However, it will not receive an increment from the second of the two updates in Step 3(e), because $y_i + y_j + r_{ij} < w_e$ and so its final count will be at most $4$. The complexity of the protocol follows from the complexity for \finv, \subgraph and the matching verification described in Section \ref{sec:mcmb}. 
\end{proof}

\input{mst-appendix}

\section{ Streaming Interactive Proofs for Approximate Metric TSP}
We can apply our protocols to another interesting graph streaming problem: that of computing an approximation to the min cost travelling salesman tour. The input here is a weighted complete graph of distances.

We briefly recall the Christofides heuristic: compute a MST $T$ on the graph and add to $T$ all edges of a
min-weight perfect matching on the odd-degree vertices of T. The classical Christofides result shows that the sum of the costs of this MST and induced min-weight matching is a $3/2$ approximation to the TSP cost.
In the SIP setting, we have protocols for both of these problems. The difficulty however is in the dependency:
the matching is built on the odd-degree vertices of the MST, and this would seem to require the verifier
to maintain much more states as in the streaming setting. We show that this is not the case, and in fact we
can obtain an efficient SIP for verifying a $(3/2 + \epsilon)$-approximation to the TSP. 

Assume $T$ is a $(1+\eps)$ approximate MST provided by the prover in the
verification protocol and let $T^*$ be the optimum MST on $G$. Also, let $A$ be
the optimum solution to TSP. Since graph $G$ is connected, we have $w(A) \geq
w(T^*)$ and because $(1+\epsilon) \cdot w(T^*) \geq w(T)$, thus $(1+\epsilon)
\cdot w(A) \geq w(T)$. Further, let $M$ be the min-cost-matching over the odd
degree set $O$. By a simple reasoning, we can show that $w(M) \leq
\frac{w(A)}{2}$, thus $w(M)+w(T) \leq (1+\epsilon) \cdot w(A) + \frac{w(A)}{2}$
and from the triangle inequality it follows that the algorithm can verify the
TSP cost within ($\frac{3}{2} + \epsilon$)-approximation.

We use first the protocol for verifying approximate MST described in
Section\ref{sec:mst}.
What remains is how we verify a min-cost perfect matching on the odd-degree nodes of the spanning tree. We employ the procedure described in Section \ref{sec:mwmg} for maximum weight matching along with a standard equivalence to min-cost perfect matching. In addition to validating all the LP constraints, we also have to make sure that they pertain solely to vertices in ODD. We do this as above by using the fingerprint for ODD to ensure that we only count satisfied constraints on edges in ODD. 

\input{tsp-appendix}

\begin{theorem}\label{TSP}
Given a weighted complete graph with $n$ vertices, in which the edge weights satisfy the triangle inequality, there exists a streaming interactive protocol for verifying optimal TSP cost within ($\frac{3}{2} + \epsilon$)-approximation with $(\log n)$ rounds of communication, and $(\log^2 n, n \log^2 n /\eps)$ cost.
\end{theorem}
\input{bhh}
\input{constant-appendix}
\input{Discussion}
\bibliographystyle{acm}
\bibliography{paper141}

\end{document}

%% file: triangles-appendix.tex
\section{Warm-up: Counting Triangles}
\label{sec:triangles}
The number of triangles in a graph is the number of induced subgraphs isomorphic
to $K_3$. Here we present a protocol to verify the number of triangles in a graph
presented as a dynamic stream of edges. We will assume that at the end of the
stream no edge has a net frequency greater than $1$.
\begin{enumerate}
\item \textsf{V} processes the input data stream consisting of tuples $(i,j,\Delta)$
  representing edges in the graph for $F^{-1}_3$  with respect to a
  vector $\mathbf{a}$ indexed by entries from $\mathcal{U} = \{
  (i,j,k) \mid i,j,k \in [n], i < j < k\}$. For each edge $e = (i, j,\Delta), i < j$ in the
  stream, \textsf{V} increments all entries $\mathbf{a}[(i, *, j)], \mathbf{a}[(*,
  i, j)]$ and $\mathbf{a}[(i, j, *)]$ by $\Delta$. Note that the input to $F^{-1}_3$ protocol is in fact these derived incremental updates from the original stream of edges in the input graph and not the tuples $(i,j,\Delta)$.
\item \textsf{P} sends the claimed value $c^*$ as the number of triangles in $G$. 
\item \textsf{V} checks the the correctness of the answer by running the verification protocol for $F^{-1}$ and checks if $F^{-1}_3  = c^*$.
\end{enumerate}
\begin{lemma}
  The above protocol correctly verifies (with a constant probability of error) the number of triangles in a graph with cost $(\log^2 n, \log^2 n)$.
\end{lemma}
\begin{proof}
  Follows from Lemma~\ref{lemma:finv} and observation that maximum
  frequency of any entry in $\mathbf{a}$ is $3$. 
\end{proof}

\paragraph*{Verifier Update Time.}
Note that while this protocol and the other graph protocols which follows achieves very small space and communication costs, but the update time could be high (polynomial in $n$) since processing a single stream token may trigger updates in many entries of $\mathbf{a}$. But by using a nice trick found in \cite{chakrabartiverifiable}, the verifier time can be reduced to poly$\log n$. Here we state the main results which can be applied to all the protocols in this paper to guarantee poly$\log n$ verifier update time.
\begin{lemma}\label{verifier-time}
Assume a data stream $\tau$ in which each element triggers updates on multiple entries of vector $\mathbf{a}$, and each entry in this vector is indexed by a multidimensional vector with $b$ coordinates and let $\mathcal{U} \subseteq {[n^c]}^b$. In all the SIP protocols for graph problems in this paper, the updates in the form of $\mathbf{a}[(\beta_1, \beta_2, \cdots, \beta_q, *, \cdots, *)]$ (which is interpreted as: update all entries $\beta$ where $\beta \in \{(\beta_1, \cdots, \beta_q, s_1, \cdots, s_{b-q}) | s_i \in [n^c], i \in [b-q], (\beta_1, \cdots, \beta_q, s_1, \cdots, s_{b-q}) \in \mathcal{U} \} $) can be done in poly$\log n$ time.
\end{lemma}


%% file: verifier-time.tex
Here we present the proof for Lemma \ref{verifier-time}. The main ideas are extracted from \cite{chakrabartiverifiable}, in which this trick is used for reducing verifier time in Nearest Neighbor verification problem. For more details, refer to Section 3.2 in \cite{chakrabartiverifiable}.
\begin{proof}
\label{verify-appendix}
 Suppose the boolean function $\phi$ which takes two vectors $\boldsymbol{\beta}$ and $\mathbf{x}$ as inputs, in which $\boldsymbol{\beta}=(\beta_1, \cdots, \beta_b)$ is a vector with $b$ coordinates each $\beta_i \in {[n]}^c$ and $\mathbf{x}=(x_1, \cdots, x_q)$ is a vector with $q <b$ coordinates each $x_i \in {[n]}^c$. Here we assume $\boldsymbol{\beta}$ is an index in the vector $\mathbf{a}$ defined over the input stream and $\mathbf{x}$ the update vector defined by the current stream element(i.e. specifies which indices in $\mathbf{a}$ must be updated). Define $\phi(\boldsymbol{\beta}, \mathbf{x} ) = 1 \leftrightarrow \beta_i = x_i , 1\leq \forall i\leq q$ with $O(\log n)$-bits inputs (since we can assume $b$ as a small constant). Let define the length of the shortest de Morgan formula for function $\phi$ as fsize$(\phi)$. Obviously, the function $\phi$ is essentially the equality check on O(logn)-bits input and we know that the addition and multiplication of $s$-bits inputs can be computed by Boolean circuits in depth $\log s$, resulting Boolean formula of size poly$(s)$. Thus, fsize$(\phi)=$ poly$\log n$. Considering the boolean formula for $\phi$, we associate a polynomial $\tilde{G}$ with each gate $G$ of this formula, with input variables $W_1, \cdots, W_{b \log n}$ and $X_1, \cdots, X_{q \log n}$, as follows:
 \begin{align*}
 G= \beta_i &\Rightarrow \tilde{G}= W_i\\
 G=x_i &\Rightarrow \tilde{G}= X_i\\
 G = \neg G_1 &\Rightarrow \tilde{G}= -\tilde{G}_1\\
 G = G_1 \wedge G_2 &\Rightarrow \tilde{G}= \tilde{G_1}\tilde{G_2}\\
 G= G_1 \vee G_2 &\Rightarrow \tilde{G}= 1-(1-\tilde{G_1}(1-\tilde{G_2}))
 \end{align*}
 Let $\tilde{\phi}(W_1, \cdots, W_{b \log n}, X_1, \cdots, X_{q \log n})$ to be the polynomial associated with the output gate, which is in fact the standard arithmetization of the formula. We consider $\tilde{\phi}$ as a polynomial defined over $\field[W_1, \cdots, W_{b \log n}, X_1, \cdots, X_{q \log n}]$ for a large enough finite field $\field$. By construction, $\tilde{\phi}$ has total degree at most fsize$(\phi)$ and agree  with $\phi$ on every Boolean input. Define the polynomial $\Psi (W_1, \cdots, W_{b \log n}) = \Sigma_{i=1} \tilde{\phi}((W_1, \cdots, W_{b \log n}), \mathbf{x}^{(i)})$, in which $\mathbf{x}^{(i)}$ is the update vector defined by the element $i$ in the stream.
 Now we can observe that the vector $\mathbf{a}$ defined by the stream updates, can be interpreted as follows:
 \begin{align*}
 a[\boldsymbol{\beta}] = \Sigma_{i=1} \phi(\boldsymbol{\beta},{\mathbf{x}}^{(i)})= \Sigma_{i=1} \tilde{\phi}(\boldsymbol{\beta},{\mathbf{x}}^{(i)})= \Psi(\boldsymbol{\beta})
 \end{align*}
 It follows that $\Psi$ is the extension of $\mathbf{a}$ to $\field$ with degree
 equal to fsize($\phi$) and can be defined implicitly by input stream. Also, the
 verifier can easily evaluate $\Psi(\mathbf{r})$ for some random point
 $\mathbf{r} \in \field^{b \log n}$, as similar to polynomial evaluation in
 \sumcheck protocol. Considering that fsize$(\phi)=$ poly$\log n$, the
 complexity result of update time follows. Note that this approach adds an extra space cost fsize$(\phi) = $poly $\log n$ for the size of Boolean formula, but in general this does not affect the total space cost of the protocols discussed in this paper.
 \end{proof}

%% file: mwmb.tex
\section{SIP for MAX-WEIGHT-MATCHING in Bipartite Graphs}
\label{sec:maxWBipartite}

Consider now a bipartite graph with edge weights, with the goal being to compute a matching of maximum weight (the weight of the matching being the sum of the weights of its edges). Our verification protocol will introduce another technique we call ``flattening'' that we will exploit subsequently for matching in general graphs. 

Recall that we assume a ``dynamic update'' model for the streaming edges: each edge is presented in the form $(e, w_e, \Delta)$ where $\Delta \in \{+1, -1\}$. Thus, edges are inserted and deleted in the graph, but their weight is not modified. We will also assume that all weights are bounded by some polynomial $n^c$. 

As before, one part of the protocol is the presentation of a matching by \textsf{P}: the verification of this matching follows the same procedure as in Section~\ref{sec:verifying-matching} and we will not discuss it further. We now focus on the problem of certifying \emph{optimality} of this matching. 

For this goal, we proceed by the standard LP-duality for bipartite maximum weight matching. Let the graph be $G = (V,E)$ and $A$ is its incidence matrix (a matrix in $\{0 , 1 \}^{V \times E}$ where $a_{ij} = 1$ iff edge $j$ is incident to vertex $i$). Let $\delta(v)$ denote the edge neighborhood of a vertex $v$ and $P_{\text{match}}$ represent the convex combination of all matchings on $G$, and note that for a bipartite graph:

\begin{equation}
P_{\text{match}} = \left \{ x \in \reals^{E}_+ : \forall v \in V, 
\sum_{e \in \delta(v)} x_e \leq 1 \right \}
\end{equation}

 Applying the LP duality theorem to the bipartite max-weight matching problem on $G$, and letting $w$ be the weight vector on the edges, we see that: 

\begin{align*}
\max \{ w^T x : x \in P_{\text{match}} (G)  \}  
	&= \max \left \{ w^T x : x \geq 0 \text{ and } \forall v \in V ,
	\sum_{e \in \delta(v)} x_e \leq 1 \right \} \\
	&= \max \left \{ w^T x: x \geq 0, Ax \leq 1 \right \} \\
	&=  \min \left \{ 1^T y : A^T y \geq w, y \geq 0 \right \} \\
	&= \min \left \{1^T y : y \geq 0 \text{ and } \forall e_{i,j} \in E, 
y_i + y_j \geq w_{i,j} \right \}
\end{align*}

Considering this formulation, a certificate of optimality for a maximum weight matching of cost $c^*$ is an assignment of weights $y_i$ to vertices of $V$ such that $\sum y_i = c^*$ and for each edge $e = (i,j), y_i + y_j \ge w_e$. 

A protocol similar to the unweighted case would proceed as follows: \textsf{P} would send over a stream $(i, y_i)$ of vertices, and the verifier would treat these as decrements to a vector over edges. \textsf{V} would then verify that no element of the vector had a value greater than zero. However, by Lemma~\ref{lemma:finv}, this would incur a communication cost linear in the maximum weight (since that is the  maximum value of an element of this vector), which is prohibitively expensive. 

The key is to observe that the communication cost of the protocol depends linearly on the maximum value of an element of the vector, but only \emph{logarithmically} on the length of the vector itself. So if we can ``flatten'' the vector so that it becomes larger, but the maximum value of an element becomes smaller, we might obtain a cheaper protocol. 

Let $\mathbf{a}$ be indexed by elements of  $\mathcal{U} = \{ ((i,j), w, y_i, y_j) \mid (i,j) \in E, w, y_i, y_j \in [n^c], i < j, w \le y_i + y_j\}$. $|\mathcal{U}| = O(n^{3c+2})$. The protocol proceeds as follows.

Intuitively, each entry of $\mathbf{a}$ corresponds to a valid dual constraint. When \textsf{V} reads the input stream of edges, it will increment counts for all entries of $\mathbf{a}$ that \emph{could} be part of a valid dual constraint. Correspondingly, when \textsf{P} sends back the actual dual variables, \textsf{V} updates all compatible entries.

\begin{enumerate}
\item \textsf{V} processes the input edge stream for $F^{-1}_3$ (with respect to $\mathbf{a}$). 
\item Upon seeing  $(e, w_e, \Delta)$ in the stream, \textsf{V} accordingly updates all entries $\mathbf{a}[(e, w, *, *)]$ by $\Delta$.
\item \textsf{P} sends a stream of $(i, y_i)$ in increasing order of $i$. 
\item \textsf{V} verifies that all $i \in [n]$ appear in the list. For each $i$, it increments all entries  $\mathbf{a}[((i,*), *, y_i, *)]$ and $\mathbf{a}[(*, i), *, *, y_i)]$.
\item \textsf{V} verifies that  $F^{-1}_3(\mathbf{a}) = m$ and accepts.
\end{enumerate}

\paragraph*{Correctness.}

Suppose the prover provides a valid dual certificate satisfying the conditions for optimality. Consider any edge $e = (i,j)$, the associated dual variables $y_i, y_j$ and the entry $r = (e, w_{ij}, y_i, y_j)$. When $e$ is first encountered, \textsf{V} will increment $a[r]$. When \textsf{P} sends $y_i$, $r$ will satisfy the compatibility condition and $a[r]$ will be incremented. A similar increment will happen for $y_j$. Note that no other stream element will trigger an update of $a[r]$. Therefore, every satisfied constraint will yield an entry of $\mathbf{a}$ with value $3$. 

Conversely, suppose the constraint is not satisfied, i.e $y_i + y_j < w_{ij}$. There is no corresponding entry of $\mathbf{a}$ to be updated in this case. This proves that the number of entries of $\mathbf{a}$ with value $3$ is exactly the number of edges with satisfied dual constraints. The correctness of the protocol follows. 

\paragraph*{Complexity.}
The maximum frequency in $\mathbf{a}$ is at most $3$ and the domain size $u = O(n^{3c + 2})$. Note that this is in contrast with the representation first proposed that would have domain size $n^2$ and maximum frequency $O(n^c)$. In effect, we have \emph{flattened} the representation. Invoking Lemma \ref{lemma:finv}, as well as the bound for verifying the matching from Section~\ref{sec:mcmb}, we obtain the following result. 

\begin{theorem}\label{thm:maxBW}
Given a bipartite graph with $n$ vertices and edge weights drawn from $[n^c]$ for some constant $c$, there exists a streaming interactive protocol  for verifying the maximum-weight matching with $\log n$ rounds of communication, space cost $O(\log^2 n)$ and communication cost $O(n \log n)$.
\end{theorem}

We can make a small improvement to Theorem \ref{thm:maxBW}. First note that the prover need only send the non-zero $y_i$ in ascending order along with label to the verifier, who can implicitly assign $y_j = 0$ to all absent weights. This then reduces the communication to be linear in the \emph{cardinality} and thereby also the cost of the maximum weight matching. Namely, we now have:

\begin{theorem}\label{thm:maxBWImp}
Given an input bipartite graph with $n$ vertices and edge weights drawn from $[n^c]$ for some constant $c$, there exists a streaming interactive protocol  for verifying the maximum-weight matching with $\log n$ rounds of communication, space cost $O(\log^2 n)$ and communication cost $O(c^* \log n)$, where $c^*$ is the cardinality of the optimal matching over the input.
\end{theorem}

\paragraph*{Note.}
We assume that \textsf{V} knows the number of edges in the graph. This assumption can be dropped easily by merely summing over all updates $\Delta$. Since we assume that every edge will have a final count of $1$ or $0$, this will correctly compute the number of edges at the end of the stream. 



%% file: mst-appendix.tex
\section{Streaming Interactive Proofs for Approximate MST}
\label{sec:mst}
For verifying the approximate weight of MST, we follow the reduction to the problem of counting the number of connected components in graphs, which was initially introduced in \cite{chazelle2005approximating} and later was generalized to streaming setting \cite{ahn2012analyzing}. Here is the main results which we use here:
\begin{lemma}[\cite{ahn2012analyzing}]
Let $T$ be a minimum spanning tree on graph $G$ with edge weights bounded by $W=$poly$(n)$ and $G_i$ be the subgraph of $G$ consisting of all edges whose weights is at most $w_i = {(1+\epsilon)}^i$ and let $cc(H)$ denote the number of connected components of graph $H$. Set $r= \lfloor \log_{1+\epsilon} W\rfloor$. Then,
\begin{align*}
w(T) \leq n-{(1+\epsilon)}^r + \sum_{i=0}^r \lambda_i cc(G_i) \leq (1+\epsilon) w(T)
\end{align*}
where $\lambda_i = {(1+\epsilon)}^{i+1} - {(1+\epsilon)}^i$.
\end{lemma}

Based on this result, we can design a SIP for verifying the approximate weight
of minimum spanning tree using a verification protocol \ref{thm:numconn} for number of connected components in a graph.
\begin{theorem}\label{thm:numconn}
Given a weighted graph with $n$ vertices, there exists a SIP protocol for verifying the number of connected components $G_i$ with $(\log n)$ rounds of communication, and $(\log^2 n, n \log n)$ cost.
\end{theorem}
\begin{corollary}
Given a weighted graph with $n$ vertices, there exists a SIP protocol for verifying MST within ($1+ \epsilon$)-approximation with $(\log n)$ rounds of communication, and $(\log^2 n, n \log^2 n / \epsilon)$ cost.
\end{corollary}
\begin{proof}
 As the verifier processes the stream, each edge weight is snapped to the closest power of $(1+\epsilon)$. Note that given an a priori bound $n^c$ on edge weights, $G$ can be partitioned into at most $\frac{\log n}{\eps} $ graphs $G_i$. We run this many copies of the connected components protocol in parallel to verify the values of $cc(G_i), \forall i$. 
\end{proof}
We now present the proof of theorem \ref{thm:numconn}. For simplicity, consider $V = (V_1 \cup \cdots \cup V_r)$ as the $r$ connected components and $T = (T_i \cup \cdots \cup T_r)$ as $r$ spanning trees on $r$ corresponding connected components, provided by prover as the certificate.
Now verifier needs to check if the certificate $T$ is valid by considering the following conditions:
\begin{enumerate}
\item \disjoint : All the spanning trees are disjoint, i.e. $T_i \cap T_j = \emptyset$ for all the pairs of trees in $T$.
\item \subgraph : Each spanning tree in $T = (T_1 \cup \cdots \cup T_r)$ is a subgraph of the input graph $G$. This may be handled by \subgraph protocol described before in Lemma \ref{lem:subgraph}.
\item \spanningtree : Each component in $T = (T_1 \cup \cdots \cup T_r)$  is in fact a \emph{spanning tree}.
\item \maximality : Each component in $T = (T_1 \cup \cdots \cup T_r)$ is in fact maximal, i.e. there is no edge between the components in original graph $G$.
\end{enumerate}

We assume that the certificate $T$ is sent by the prover in streaming manner in the following format and both players agree on this at the start of the protocol:
\begin{align*} 
T: \{|T|, r, (T_1 , \cdots , T_r) \}
\end{align*}
For the representation of spanning trees, we consider a topological ordering on each tree $T_i$, starting from root node $root_i$, and each directed edge $(v_{out}, v_{in})$ connects the parent node $v_{out}$ to the child node $v_{in}$:
\begin{align*}
T_i: \{root_i, \cup e (v_{out}, v_{in})\}
\end{align*}

Here we present the protocols for checking each of these conditions. The following \disjoint protocol will be called as a subroutine in our \spanningtree main protocol.

\paragraph*{Protocol: \disjoint}
\begin{enumerate}
\item \textsf{P} sends over $r$ components of $T_i$ in $T$ in streaming manner . 
\item \textsf{P} ``replays" all the edges $T'$ in the tuple form $(e_{i,j}, \ell)$ for $i$, $j \in [n]$ and $\ell \in [r]$ denotes the component $e_{i,j}$ is assigned to. The edges in $T'$ are presented according to a canonical total ordering on the edge set, and hence $\textsf{V}$ can easily check that $T'$ has no repeated edge; i.e. the same edge presented in two distinct components.
\item Fingerprinting can then be used to confirm that $T' = \cup_i T_i$ with high probability, and hence that each edge occurs in at most one tree.
\item A similar procedure is run to ensure that no vertex is repeated in more than one $T_i$ and that every vertex is seen at least once.
\end{enumerate}

To check if each component in the claimed certificate $T$ sent by the prover is in fact spanning tree, the verifier needs to check that each $T_i$ is \emph{cycle-free} and also \emph{connected}. For this goal we present the following protocol: 
\paragraph*{Protocol: \spanningtree}
\begin{enumerate}
\item \textsf{P} sends over the certificate $T: \{|T|, r, (T_1 , \cdots , T_r) \}$ in which each $T_i$ is of the form  $\{i, \cup e (v_{out}, v_{in})\}$ and the root $u_i$ of $T_i$ is presented first.
\item \textsf{V} runs the \disjoint protocol to ensure that in the certificate $T: \{|T|, r, (T_1 , \cdots , T_r) \}$ all  $T_i$ and $T_j$ are edge and vertex disjoint for all $i \neq j$.
\item \textsf{V} has the prover again similarly replay $\cup_i T_i$ ordered by the label of the in-vertex of each edge to ensure that each vertex except the root $u_i$ has exactly one incoming edge; i.e. that all $T_i$ are cycle free and connected.
\end{enumerate}
We now need to check that there is no edge between sets $V_i$ and $V_j$ for $i \neq j$:
\paragraph*{Protocol: \maximality}
\begin{enumerate}
\item We define an extended universe $U$ now of size $n^3$, with elements $(e_{i,j}, k )$ where $k \in [n]$ represents the label of the component. 
\item \textsf{V} initiates the $F^{-1}_{-1}$ protocol on the input stream. Upon seeing any edge $e_{i,j}$, $\textsf{V}$ increments by $1$ all tuples containing $e_{i,j}$.
\item \textsf{P} sends the label of each vertex $(v_i, j)$ where $j \in [n]$ represents the label of the component in the certificate. (Note that the verifier can ensure this input is consistent with the $T = \cup_i T_i$ sent earlier by simply fingerprinting as described before).

\item \textsf{V} considers each vertex $v_i \in V_j$  as a decrement update by $2$ on all $n$ possible tuples compatible with $v_i$  and component label $j$. This step can be assumed as continuing the process for $F^{-1}_{-1}$ mentioned in the first step.

\item $F^{-1}_{-1}$ corresponds to exactly the set of edges observed in stream and crossing between two $V_i$ and $V_j$ for $i \neq j$. To see this, we enumerate the cases explicitly:
\begin{enumerate}
\item $\{-3, -4 \}$ are the possible values for an edge $(e_{a,b}, i)$ with both endpoints contained in a single $V_i$ corresponding to whether $e_{a,b}$ was originally in the stream or not. (The edge is decremented twice by $2$ in the derived stream.)
\item $\{-1,-2 \}$ are the possible values for any edge $(e_{a,b}, i)$ and $(e_{a,b}, j)$ with one endpoint in a $V_i$ and the other in $V_j$ for $i \neq j$, corresponding to whether $e_{a,b}$ was originally in the stream or not. ($e_{a,b}$ is decremented exactly once by $2$ in each of the two copies corresponding to $i$ and $j$ respectively.)
\end{enumerate}
\item \textsf{V} runs $F^{-1}$ with $\textsf{P}$ and accepts that there are no edges between the $V_i$ if and only if $F^{-1}_{-1} = 0$.
\end{enumerate} 
\paragraph*{Complexity Analysis of the Protocol}
We know the cost of $F^{-1}_{-1}$ protocol is $(\log^2 n, \log^2 n)$ for frequency ranges bounded by a constant, whereas the costs of the remaining fingerprinting steps and sending the certificate are at most  $(\log n, n \log n)$.  Hence the cost of our protocol is dominated by $(\log^2 n, n \log n )$ in the worst case. The verifier update cost on each step is bounded as $O(n^2)$.

\paragraph*{Testing Bipartiteness}
As a corollary of Theorem \ref{thm:numconn} it is also possible to test whether a graph is bipartite. This follows by applying the connectivity verification protocol described before on the both input graph $G$ and the \emph{bipartite double cover} of $G$, say $G'$. The bipartite double cover of a graph is formed by making two copies $u_1, u_2$ of every node $u$ of $G$ and adding edges $\{u_1, v_2\}$ and $\{u_2, v_1\}$ for every edge $\{u, v\}$ of $G$. It can be easily shown that $G$ is bipartite if and only if the number of connected components in the double cover $G'$ is exactly twice the number of connected components in $G$.
\begin{corollary}\label{bipartite}
Given an input graph $G$ with $n$ vertices, there exists a SIP protocol for testing bipartiteness on $G$ with $(\log n)$ rounds of communication, and $(\log^2 n, n \log n)$ cost.
\end{corollary}

\paragraph*{Remark}We note here that while we could have used known parallel algorithms for connectivity and MST combined with the protocol of Goldwasser et al \cite{GKR} and the technique of Cormode, Thaler and Yi \cite{cormode2011verifying} to obtain similar results, we need an explicit and simpler protocol with an output that we can fit into the overall TSP protocol described later in next section.

%% file: tsp-appendix.tex
Here we describe the protocol for verifying approximate metric TSP in full details, which results in Theorem \ref{TSP}:

\subsection{The TSP Verification Protocol.}
\begin{enumerate}
\item \textsf{P} presents a spanning tree which is claimed to be MST and can be verified within $(1+\epsilon)$-approximation by \textsf{V} (as described in Section \ref{sec:mst}). $\textsf{V}$ maintains a fingerprint on the vertices by using the \mse algorithm and updating the frequency of each vertex seen as an endpoint of an edge in the tree. This results in a fingerprint where each vertex has multiplicity equal to its degree in the MST. 
\item \textsf{P} then lists all vertices of the spanning tree in lexicographic order annotated with their degree. \textsf{V} verifies that this fingerprint matches the one constructed in the previous step and builds a new fingerprint for the set $ODD$ of all odd-degree vertices (disregarding their degree).
\item \textsf{P} presents a claimed min-weight perfect matching on the vertex set $ODD$ 
\item \textsf{V} verifies that this list of edges is indeed a matching using the protocol from Section \ref{sec:mcmb}. In addition, it verifies that the vertices touched by these edges comprise $ODD$ by using \mse to validate the fingerprint from the previous step. 
\item To verify the lower bound on min-weight perfect matching, we first reduce to max-weight matching. Let $W = n^c$ be the a priori upper bound on the weight of each edge. Replace all weights $w$ by $W+1 - w$. Clearly now on a complete graph the max-weight matching corresponds to the min-weight perfect matching. 
\item First, \textsf{V} needs to ensure that the entire certificate $C$ is contained inside the $ODD$ set. Recall that \textsf{V} has maintained a fingerprint of $ODD$, so we may use a variant of \mse. \textsf{P} replays $C$, along with any vertices which are in $ODD$ but not in $C$ and $\textsf{V}$ checks the fingerprints match.
\item Then, \textsf{V} needs to check that all the constraints for the problem is satisfied by the certificate. This step is identical to what we described before for Maximum-Weight-Matching~\ref{sec:mwmg}(counting the ``good'' tuples), but here the satisfied constraints must be counted only on $ODD$ set. 
\item For this goal, we amend the protocol of Section~\ref{sec:mwmg} so that \textsf{P} streams the subset $V-ODD$ to the verifier and then  \textsf{V} can simply decrement the frequency of all the tuples defined on $V-ODD$ by 1. Now all tuples corresponding to edges not containing both endpoints in $ODD$ may achieve frequency at most $4$ and hence will not be counted by the $F^{-1}_5$ query. 
\item Again, the accuracy of the claimed $V-ODD$ can be checked by using \mse. Let $D$ be the claimed $V-ODD$.  \textsf{P} streams $D$ to $\textsf{V}$, who checks by \mse that the fingerprint of $D \cup ODD$ matches that of the entire vertex set. (Note that fingerprints are linear, so the fingerprint of $D \cup ODD$ is just the fingerprint of $D$ plus the fingerprint of $ODD$.)
\item Now, \textsf{V} accepts the max-weight matching certificate if and only if the number of ``good'' tuples (which determines the count of satisfied edge constraints) is $|ODD| \choose 2$ (i.e. the number of edges in complete graph induced by the $ODD$ set). As discussed earlier, these correspond to the value of $F^{-1}_5$ in our extended universe.
\end{enumerate}

 Finally, the approximate TSP cost is the sum of the min-weight perfect matching on $ODD$ and the MST cost on the graph.


%% file: bhh.tex
\section{Boolean Hidden Hypermatching and Disjointness}
\label{sec:bool-hidd-hyperm}
Boolean Hidden Matching ($BHH_n^t$)  is a two-party one-way communication problem in which Alice's input is a boolean vector $x \in {\{0, 1\}}^n$ where $n = kt$ for some integer $k$ and Bob's input is a (perfect) hypermatching $M$ on the set of coordinates $[n]$, where each edge $M_r$ contains $t$ vertices represented by indices as $\{M_{r,1}, ..., M_{r,t}\}$, and a boolean vector $w$ of length $\frac{n}{t}$. We identify the matching $M$ with its edge incidence matrix. Let $Mx$ denote the length $\frac{n}{t}$ boolean vector $(\bigoplus_{1\leq i \leq t} x_{M_{1, i}}, \cdots , \bigoplus_{1\leq i \leq t} x_{M_{\frac{n}{t}, i}})$. It is promised in advance that there are only two separate cases:

YES case: The vector $w$ satisfies $Mx \bigoplus w = 0^{\frac{n}{t}}$.

NO case: The vector $w$ satisfies $Mx \bigoplus w = 1^{\frac{n}{t}}$.

The goal for Bob is to differentiate these two cases. 

The following lower bound result for $BHH_n^t$ is obtained in \cite{verbin2011streaming}:
\begin{lemma}(\cite{verbin2011streaming})
Any randomized one-way communication protocol for solving $BHH_n^t$ when $n=kt$ for some integer $k$, with error probability at most $\frac{1}{4}$ requires $\Omega(n^{1-\frac{1}{t}})$ communication.
\end{lemma}

\begin{lemma}\label{bhh-ver}
Consider the streaming version of $BHH_n^t$ problem, in which the binary vector $x$ comes in streaming, followed by edges in $M$ along with the boolean vector $w$ for weights. There exists a streaming interactive protocol with communication and space cost $O(t \cdot \log n (\log \log n))$ for $BHH_n^t$ problem.
\end{lemma}

\begin{proof}
Considering the promise that we have in YES and NO case of $BHH_n^t$ communication problem, it is enough to query the weights of vertices on only one of the hyperedges on the matching and compare it to the corresponding weight in vector $w$. This way the $BHH_n^t$ problem can be reduced to $t$ instances of INDEX problem. Assume the vector $x$ as the input stream and take one of the followed hyperedges, say $M_r = \{M_{r,1}, ..., M_{r,t}\}$, as the $t$ query index. In this scenario the verifier just need to apply the verification protocol INDEX in $t$ locations $\{M_{r,1}, ..., M_{r,t}\}$ on $x$ and check if $\bigoplus_{1\leq i \leq t} x_{M_{r, i}} \bigoplus w_r = 0$ or $\bigoplus_{1\leq i \leq t} x_{M_{r, i}}  \bigoplus w_r = 1$. According to \cite{chakrabartiverifiable}, the verification (communication and space) cost for INDEX problem is $O(\log n (\log \log n))$ and this results in $O(t \cdot \log n (\log \log n))$ cost for $BHH_n^t$.
\end{proof}

 We now show a similar result for Disjointness($DISJ_n$).  $DISJ_n$ is a two-party one-way communication problem in which Alice and Bob each have a boolean vector $x$ and $y \in {\{0, 1\}}^n$ respectively, and they wish to determine if there is some index $i$ such that $a_i = b_i = 1$. Razoborov~\cite{disjhardness} shows an $\Omega(n)$ lower bound on the communication complexity of this problem for one-way protocols. We show now however that $DISJ_n$ is easy in the SIP model.

\begin{lemma}\label{dis-ver}
Consider the streaming version of $DISJ_n$ problem, in which the binary vector $x$ comes in streaming, followed by binary vector $y$. There exists a streaming interactive protocol with communication and space cost $O(\log^2 n)$ for $DISJ_n$.
\end{lemma}

\begin{proof}
The verifier maintains a universe $\mathcal{U}$ corresponding to $[n]$. When the verifier sees the $i^{th}$ bit of $x$, it increments the frequency of universe element $i$ by $x_i$. Now when the verifier streams $y_i$, the verifier again increments the frequency of element $i$ by $y_i$. Clearly, $x$ and $y$ correspond to disjoint sets if and only if $F^{-1}_{2} = 0$. We can then simply run the \finv protocol, and the bound follows by Lemma~\ref{lemma:finv}
\end{proof}

Lemma \ref{bhh-ver} and \ref{dis-ver} shows that while $BHH_n^t$ and $DISJ_n$ are lower bound barriers to computations in the streaming model, however they are easily tractable in the streaming verification setting. This gives a first suggestion that for problems such as MAX-CUT and MAX-MATCHING where most of the known lower bounds go through $BHH_n^t$ or $DISJ_n$, streaming verification protocols may prove more effective, and was the initial motivation for our study.

%% file: constant-appendix.tex
\section{Revisit the sum-check protocol with constant rounds}
\label{sec:constant-round}
As mentioned before, we have a $(\log u)$-rounds $(\log u)$-cost verification protocol for any frequency-based functions by applying the sum-check as the core of the protocol. In this section, we study on the possibility of reducing the round-complexity of these protocols to constant-rounds.

For this goal, we revisit the sum-check protocol described in \cite{cormode2011verifying} and briefly explain the details of the protocol and the complexity analysis.

The sum-check which we present here is for verifying $F_1(\mathbf{a}) = \sum_{i \in [u]} f_{\mathbf{a}} (i) = \sum_{x_1, \cdots, x_d \in {[\ell]}^{d}} f_{\mathbf{a}}(x_1, x_2, \cdots, x_d)$, in which $u = {[\ell]}^d$. We can simply extend this protocol to any frequency-based function defined as $F(\mathbf{a}) = \sum_{i \in [u]} h(\mathbf{a}_i) = \sum_{i \in [u]} h \circ f(i)$. 

We briefly describe the construction of this extension polynomial. Start from $f_{\mathbf{a}}$ and rearrange the frequency vector $\mathbf{a}$ into a $d$-dimensional array in which $u = {\ell}^d$ for a choosen parameter $\ell$. This way we can write $i \in [u]$ as a vector $({(i)}_1^\ell, ..., {(i)}_d^\ell) \in {[\ell]}^d$. Now we pick a large prime number for field size $|\field| > u$ and define the \emph{low-degree extension} (LDE) of $\mathbf{a}$ as $f_a(\mathbf{x}) = \sum_{\mathbf{v} \in {[\ell]}^d} a_{\mathbf{v}} \chi_{\mathbf{v}}(\mathbf{x})$, in which $\chi_{\mathbf{v}}(\mathbf{x}) = \prod_{j=1}^d \chi_{v_j}(x_j)$ and $\chi_k(x_j)$ has this property that it is equal to $1$ if $x_j = k$ and $0$ otherwise. This indicator function can be defined by Lagrange basis polynomial as follows:
\begin{equation}
 \frac{(x_j-0) ... (x_j - (k-1))(x_j - (k+1)) ... (x_j - (\ell -1))}{(k-0) ... (k-(k-1))(k-(k+1)) ... (k-(\ell-1))} \label{eq:chi}
\end{equation}

Observe that for any fixed value $\mathbf{r} \in {[\field]}^d$, $f_a(\mathbf{r})$ is a linear function of $\mathbf{a}$ and can be evaluated by a streaming verifier as the updates arrive. This is the key to the implementation of the sum-check protocol. 

At the start of protocol, before observing the stream, \textsf{V} picks a random point, presented as $\mathbf{r}= (r_1, \cdots, r_d) \in {[\field]}^d$ in the corresponding field. Then computes $f_{\mathbf{a}}(\mathbf{r})$ incrementally as reads the stream updates on $\mathbf{a}$. After observing the stream, the verification protocol proceeds in $d$ rounds as follows:

In the first round, \textsf {P} sends a polynomial $g_1(x_1)$, claimed as :
\begin{align*}
g_1(X_1) = \sum_{x_2, \cdots, x_d \in {[\ell]}^{d-1}} f_{\mathbf{a}}(X_1, x_2, \cdots, x_d)
\end{align*}

Note that in this stage, if polynomial $g_1$ is the same as what is claimed here by \textsf{P}, then $F_1(\mathbf{a}) = \sum_{x_1 \in [\ell]} g_1(x_1)$.

Followng this process, in round $j > 1$, \textsf{V} sends $r_{j-1}$ to \textsf{P}. Then \textsf{P} sends a polynomial $g_j(x_j)$, claiming that:
\begin{align*}
g_j(X_j) = \sum_{x_{j+1}, \cdots, x_d \in {[\ell]}^{d-j}} f_{\mathbf{a}} (r_1, \cdots, r_{j-1}, X_j, x_{j+1}, \cdots, x_d)
\end{align*}

In each round, \textsf{V} does consistency checks by comparing the two most recent polynomials as follows:
\begin{align*}
g_{j-1}(r_{j-1}) = \sum_{x_j \in {[\ell]}} g_j(x_j)
\end{align*}
 Finally, in the last round, \textsf{P} sends $g_d$ which is claimed to be:
 \begin{align*}
 g_d(X_d) = f_{\mathbf{a}}(r_1, \cdots, r_{d-1}, X_d)
 \end{align*}
 
 Now, \textsf{V} can check if $g_d(r_d) = f_{\mathbf{a}}(\mathbf{r})$. If this test (along with all the previous checks) passes, then \textsc{V} accepts and convinced that $F_1(\mathbf{a}) = \sum_{x_1 \in [\ell]} g_1(x_1)$. 
 
 \paragraph{Complexity Analysis.} The protocol consists of $d$ rounds, and in each of them a polynomial $g_j$ is sent by \textsf{P}, which can be communicated using $O(\ell)$ words. This results in a total communication cost of $O(d \ell)$. $V$ needs to maintain $\mathbf{r}$, $f_{\mathbf{a}}(\mathbf{r})$ which each requires $(d+1)$ words of space, as well as computing and maintaining the values for a constant number of polynomials in each round of sum-check. As described before, this is required for comparing the two most recent polynomials by checking
 \begin{align*}
g_{j-1}(r_{j-1}) = \sum_{x_j \in {[\ell]}} g_j(x_j)
\end{align*}

Each of the $g_j$ communicated in round $j$ is a univariate polynomial with degree $(\ell - 1)$ and can be described in $(\ell - 1)$ words. Let's represent each polynomial $g_j$ as follows:

\begin{align*}
g_j(x_j) = \sum_{i \in [\ell - 1]} c_{ij} x_j^{i} 
\end{align*}
 
In each round $j$ the verifier requires to do the consistency checks over the recent polynomials as follows:
 \begin{align*}
g_j(r_j)  =  \sum_{x_j \in [\ell]} \sum_{i \in [\ell - 1]} c_{ij} x_j^{i} 
 \end{align*}
 
By reversing the ordering over the sum operation, we can rewrite this check as:
\begin{align*}
g_{j-1}(r_{j-1})  =   \sum_{i \in [\ell - 1]} \sum_{x_j \in [\ell]} c_{ij} x_j^{i}  = \sum_{i \in [\ell - 1]} c_{ij} \sum_{x_j \in [\ell]}  x_j^{i} 
\end{align*}

Let $y_i = \sum_{x_j \in [\ell]}  x_j^{i}$. Then, this will be equivalent to:
\begin{align*}
g_{j-1}(r_{j-1}) = \sum_{i \in [\ell - 1]} c_{ij} y_i
\end{align*}

Both sides in this test can be computed and maintained in $O(1)$ words space as \textsf{V} reads the polynomials $g_{j}$ presented by \textsf{P} in streaming manner. Thus, the total space required by \textsf{V} is $O(d)$ words.

By selecting $\ell$ as a constant (say 2), we obtain both space and communication cost $O(\log u)$ words for sum-check protocol which runs in $\log u$ rounds and the probablity of error is $\frac{\ell d}{|\field|}$.
 
 Now to obtain constant-rounds protocol, we can set $\ell = O(u^{\frac{1}{\gamma}})$ for some integer constant $\gamma > 1$, and considering $u = {[\ell]}^d$, we get $d = \gamma = O(1)$ (note that $d$ controls the the number of rounds), result in a protocol with constant rounds and total communication $O(u^{\frac{1}{\gamma}})$ words, while maintaining the low space cost $\gamma = O(1)$ words for \textsf{V}. The failure probability goes to $O(\frac{u^{\frac{1}{\gamma}}}{|\field|})$, which by choosing $|\field|$ larger than $u^b$ it can be made less than $\frac{1}{u^b}$ for any constant $b$ without changing the asymptotic bounds.

\paragraph{Constant round for frequency-based functions}
Here for verifying any statistic $F(\mathbf{a}) = \sum_{i \in [u]} h(\mathbf{a}_i) = \sum_{i \in [u]} h \circ f(i)$ on frequency vector $\mathbf{a}$, we use the similiar ideas to basic sum-check protocol which we described for verifying $F_1$, but the polynomials communicated by prover will be based on functions $h \circ f_{\mathbf{a}}$:

In the first round, \textsf {P} sends a polynomial $g'_1(X_1)$, claimed as:
\begin{align*}
g'_1(X_1) = \sum_{x_2, \cdots, x_d \in {[\ell]}^{d-1}} h \circ f_{\mathbf{a}}(X_1, x_2, \cdots, x_d)
\end{align*}

If polynomial $g'_1$ is the same as what is claimed here by \textsf{P}, then $F(\mathbf{a}) = \sum_{x_1 \in [\ell]} g'_1(x_1)$.

In each round $j > 1$, \textsf{V} sends $r_{j-1}$ to \textsf{P}. Then \textsf{P} sends a polynomial $g'_j(X_j)$, claimed as:
\begin{align*}
g'_j(X_j) = \sum_{x_{j+1}, \cdots, x_d \in {[\ell]}^{d-j}} h \circ f_{\mathbf{a}} (r_1, \cdots, r_{j-1}, X_j, x_{j+1}, \cdots, x_d)
\end{align*}

Again, consistency checks in each round is done by \textsf{V} by comparing the two most recent polynomials:
\begin{align*}
g'_{j-1}(r_{j-1}) = \sum_{x_j \in {[\ell]}} g'_j(x_j)
\end{align*}
And finally the verification process will be completed in the last round by sending polynomial $g'_d(X_d)$ by \textsf {P}, claimed as:
 \begin{align*}
 g'_d(X_d) = h \circ f_{\mathbf{a}}(r_1, \cdots, r_{d-1}, X_d)
 \end{align*}
 
followed by checking if $g'_d(r_d) = h \circ f_{\mathbf{a}}(\mathbf{r})$ by \textsf {V}.

\paragraph{Complexity Analysis.} Protocol consists of $d$ rounds, and in each of them a polynomial $g'_j$ with degree $O(\deg(h) \cdot \ell)$ is sent by \textsf{P}, which can be communicated using $O(\deg(h) \cdot \ell)$ words. This results in a total communication cost of $O(\deg(h) \cdot d \ell)$. \textsf{V} needs to maintain $\mathbf{r}$, $h \circ f_{\mathbf{a}}(\mathbf{r})$ (each requires $O(d)$ words space) as well as computing and maintaining the value for a constant number of polynomials in streaming manner in each round of protocol (requires $O(1)$ words of space), which results in a total space of $O(d)$ words of space. By selecting $\ell$ as a constant (say 2), then we obtain communication cost $O(\deg(h) \cdot \log u)$ and space cost $O(\log u)$ words for sum-check protocol which runs in $\log u$ rounds and the probability of error is $\frac{\deg(h) \cdot \ell d}{|\field|}$.

Note that the number of variables in input function to sum-check protocol determines the number of rounds and for any frequency-based function defined as $F(\mathbf{a}) = \sum_{i \in [u]} h(\mathbf{a}_i) = \sum_{i \in [u]} h \circ f(i)$, in which $h$ is a univariate function, the number of variables will not change and will be the same as $f_{\mathbf{a}}$. This implies that by applying the same trick as described above for reducing the number of variables in $f_{\mathbf{a}}$ (by setting $\ell = O(u^{\frac{1}{\gamma}})$ for some integer constant $\gamma > 1$ and $d = \gamma = O(1)$), we can obtain a constant-round protocol for verifying any statistics $F(\mathbf{a})$ defined by the frequency vector on the input stream, with space cost $\gamma = O(1)$ words and communication cost $O(\deg(h) \cdot u^{\frac{1}{\gamma}})$ words, while keeping the probability of error as low as $O(\frac{\deg(h)}{u^b})$ for some integer $b > 1$ (by choosing $|\field| > u^b$).

\begin{lemma}[\cite{cormode2011verifying}]
\label{lem::sum-check-protocol-cr}
There is a SIP to verify that $\sum_{i \in [u]} h(\mathbf{a}_i) = K$ for some claimed $K$,
with constant-round ($\gamma$), space cost $\gamma \cdot \log n = O(\log n)$ bits and communication cost $O(\deg(h) \cdot u^{\frac{1}{\gamma}} \cdot \log n)$ bits, while keeping the probability of error as low as $O(\frac{\deg(h)}{u^b})$ for some integer $b > 1$ (by choosing $|\field| > u^b$).
\end{lemma}

\begin{corollary}[\cite{cormode2011verifying}]
\label{constant-sumcheck}
Let $h$ be a univariate polynomial defined on the frequency vector $a$ of a graph $G$ under our model.
There is a SIP for the function $F(\tau) = \sum_{i \in [u]} h(\mathbf{a}_i)$.  The total number of rounds is constant $\gamma$ and the cost of the protocol is $\left(\gamma \log n ,   u^{\frac{1}{\gamma}} \log n\right)$.
\end{corollary}

\textbf{Remark} Note that in all the $(\log n)$-rounds verification protocols which we presented before the space cost is $\log^2 n$ bits and with changing the protocol to constant $\gamma$-rounds, we improve the space to $O(\log n)$ bits. On the other hand, in most of these protocols the communication cost is dominated by the size of \emph{certificate}, which is generally bounded by $O(n \log n)$. Thus, while using constant-round sum-check as the core of verification protocols will increase the related communication cost by a $n^{\frac{1}{\gamma}}$ factor, but that will not change the total communication cost of SIPs for matching and TSP, in which the communication cost is dominated by the certificate size.

%% file: Discussion.tex
\section{Discussion and Future Directions}
Our matching protocol requires the prover to send back an actual matching and a certificate for it. Suppose we merely wanted to verify a claimed cost for the matching. Is there a way to verify this with less communication? Another interesting question is to consider designing SIPs for graph problems which are known to be NP-hard. For example, is there any efficient SIP for verifying \textsc{Max Cut} in streaming graphs? (motivated by the fact that in standard streaming setting, even approximating maxcut is known to be hard and space lower bounds exist \cite{kapralov2015streaming}).

In our SIPs for matching, we assume that the edge weight updates are atomic. Can we relax this constraint? Justin Thaler \cite{discuss-justin} observed that by using techniques from \cite{cormode2013streaming}, we can design a SIP with $\log^2 n$ space cost and $O(W  (\log W + \log n) \log n)$ bits of communication, in which $W$ is the upper bound on the edge weights (i.e. $w_{ij} \leq W$). Both protocols will result in similar costs for any instance where the edge weights are at most $O(n)$, with having the advantage of handling incrementally-specified edge weights in the second approach. However, in the general case where $w_{ij} \in [n^c]$, our solution still has lower communication cost (worst case $n \log n$).
